\documentclass[11pt]{article}

\usepackage[margin=1.1in]{geometry}
\usepackage{amsmath,amssymb,amsthm,mathtools}
\usepackage{mathrsfs}
\usepackage{bm}
\usepackage{enumitem}
\usepackage{microtype}
\usepackage{hyperref}
\usepackage{tikz-cd}
\usepackage{amsmath,amssymb}
\usepackage{doi}

\DeclareMathOperator{\Bal}{Bal}
\DeclareMathOperator{\Hom}{Hom} 

\hypersetup{
  colorlinks=true,
  linkcolor=blue,
  citecolor=blue,
  urlcolor=blue
}

\theoremstyle{plain}
\newtheorem{theorem}{Theorem}[section]
\newtheorem{proposition}[theorem]{Proposition}

\newtheorem{corollary}[theorem]{Corollary}

\theoremstyle{definition}
\newtheorem{definition}[theorem]{Definition}
\newtheorem{assumption}[theorem]{Assumption}

\theoremstyle{remark}
\newtheorem{remark}[theorem]{Remark}

\newcommand{\QCD}{\mathrm{QCD}}
\newcommand{\DIS}{\mathrm{DIS}}
\newcommand{\OPE}{\mathrm{OPE}}
\newcommand{\MSbar}{\overline{\mathrm{MS}}}
\newcommand{\End}{\mathrm{End}}

\newcommand{\cV}{\mathcal{V}}

\newcommand{\mconv}{\mathbin{\circledast}}

\newcommand{\vC}{\bm{C}}
\newcommand{\vf}{\bm{f}}

\newcommand{\dd}{\mathrm{d}}

\title{\bf A Core Representation Theorem for Scheme-Invariant Collinear Factorization in QCD}
\author{Dustin Keller}
\date{\today}

\begin{document}
\maketitle

\begin{abstract}
Collinear factorization and the leading-twist operator product expansion (OPE) in perturbative QCD express
suitably inclusive observables in scale-separated kinematics as composites of perturbative short-distance
coefficients with universal long-distance non-perturbative correlators such as parton distribution functions (PDFs),
up to controlled power corrections. A persistent structural feature is \emph{presentation non-uniqueness}:
coefficients and correlators are not individually physical, but are defined only up to finite factorization-scheme
redefinitions induced by collinear subtractions and renormalized-operator mixing.
We formalize this redundancy categorically by introducing an \emph{interface algebra object} encoding admissible
finite collinear counterterms/mixing kernels and by organizing coefficient data and hadronic data as right/left
modules over this algebra in a symmetric monoidal category encoding the chosen recomposition calculus.
Our main result, the \emph{Core Representation Theorem}, identifies the universal scheme-invariant carrier:
the functor of balanced (scheme-invariant) pairings is represented by the relative tensor product $C\otimes_A f$,
which is terminal among all quotients of the naive composite $C\otimes f$ that preserve scheme-invariant semantics.
Finally, we show how standard physics inputs (symmetry constraints, locality/OPE, and a stated accuracy truncation)
canonically induce the interface algebra and module structures, and we prove a minimal closure principle for
completing a generating set of long-distance operators/correlators to an $A$-stable sector.
\end{abstract}


\section{Introduction}

\subsection{Factorization as a presentation of observable content}
Factorization theorems are among the most powerful structural results in perturbative QCD.
In kinematics with a hard scale $Q\gg\Lambda_{\QCD}$ and in sufficiently inclusive observables,
one can write cross sections or structure functions as composites of
(i) short-distance coefficient functions, computable in perturbation theory, and
(ii) long-distance hadronic correlators (PDFs, distribution amplitudes, \dots) which encode hadron structure
and are universal within a specified class of observables.
Canonical accounts include Collins' monograph \cite{Collins:FoPQCD}, the Collins--Soper--Sterman review \cite{CSS:Factorization},
and the PDG review of structure functions \cite{PDG2024StructureFunctions}.

A technical and conceptual difficulty is that this decomposition is not unique:
coefficients and correlators are \emph{auxiliary} constituents whose separate definitions depend on choices
of collinear subtraction, renormalization of composite operators, and operator basis in the presence of mixing.
In practice one makes (and must make) finite redefinitions
\begin{equation}\label{eq:intro-scheme}
  \vf \mapsto \vf' = Z\mconv \vf,\qquad
  \vC \mapsto \vC' = \vC\mconv Z^{-1},
\end{equation}
with an invertible matrix of kernels $Z$, leaving the recomposed observable unchanged to the stated perturbative and power accuracy
\cite{Collins:FoPQCD,PDG2024StructureFunctions}.
In moment space this is the familiar statement that operator mixing and compensating Wilson-coefficient transformations
leave physical moments invariant \cite{Wilson1969,Wilson1970,WilsonZimmermann1972,Collins:FoPQCD}.
We emphasize from the outset that this non-uniqueness is \emph{structural}: it is a built-in redundancy
forced by collinear subtractions and renormalized-operator mixing, not an ambiguity to be ``fixed'' by a preference for a particular scheme.

\subsection{Scheme-invariant semantics}
The organizing idea of this paper is to treat a factorized description as a \emph{presentation} of an underlying observable.
Fix a phase-space regime $R$ (e.g.\ leading-power collinear kinematics) and a stated accuracy $\alpha$
(e.g.\ leading twist and a perturbative truncation).
Within that fixed physics envelope, different scheme choices correspond to different presentations of the \emph{same} observable content.

Categorically, we model this as follows.
There is a ``presentation category'' of factorization data, whose objects encode the ingredients needed to write a factorized formula
(e.g.\ a coefficient object and a correlator object together with their symmetry/scale structure),
and whose isomorphisms encode admissible finite scheme changes such as \eqref{eq:intro-scheme}.
The observable map---the recomposition rule that sends a presentation to a numerical/physical prediction---is then required to be invariant
under these isomorphisms.
The \emph{scheme-invariant semantics}, is then:
\begin{quote}
  \emph{presentation-invariant observable content, i.e.\ invariance under isomorphisms in the scheme-change groupoid of presentations.}
\end{quote}
This is intentionally a 1-categorical notion of invariance (up to isomorphism).
One can imagine refinements where scheme changes form higher coherences and ``semantics'' is required to be homotopy-invariant,
but such higher-categorical structure is not needed for collinear/OPE factorization at fixed accuracy and we do not assume it.

\subsection{Parsimony as semantics-preserving quotienting and irreducibility}
The second organizing idea is a precise notion of \emph{parsimony}.
In much of the physics and machine-learning literature, ``parsimony'' is used informally as a proxy for simplicity
(e.g.\ fewer parameters, fewer terms, smaller expressions).
Here we use a structural notion:

\begin{quote}
\emph{A representation is parsimonious if it contains no internal redundancy beyond that forced by the scheme-change equivalence,
while preserving \textbf{all} scheme-invariant observable content at the chosen $(R,\alpha)$.}
\end{quote}

In the present setting, ``redundancy'' is not an aesthetic preference but an explicit equivalence generated by admissible finite redefinitions.
Our main theorem shows that this redundancy can be quotiented out \emph{universally}.
More precisely, we construct an \emph{interface algebra object} $A$ encoding admissible finite collinear counterterms/mixing kernels,
organize coefficient data as a right $A$-module and correlator data as a left $A$-module in a symmetric monoidal category $\cV$
encoding the chosen recomposition calculus, and prove that the functor of scheme-invariant pairings is \emph{representable}:
it is represented by the relative tensor product $C\otimes_A f$.

From the physics perspective, $C\otimes_A f$ is the formalization of dividing out the scheme redundancy:
it is obtained from the naive composite $C\otimes f$ by imposing the balancing relation
\begin{equation}\label{eq:intro-balance}
  (C\cdot a)\otimes f \;\sim\; C\otimes (a\cdot f),\qquad a\in A,
\end{equation}
which expresses the fact that inserting an admissible finite counterterm on the coefficient side is equivalent to inserting it on the correlator side.
For mathematicians, the construction is a coequalizer/coinvariant object (a relative tensor product) in $\cV$.

This universal property makes parsimony \emph{theorem-like} rather than heuristic.
The core $C\otimes_A f$ is terminal among all quotients of $C\otimes f$ through which every scheme-invariant evaluation factors.
In this sense it is an \emph{irreducible carrier of scheme-invariant information}:
any further quotient would necessarily identify two presentations that are distinguished by some scheme-invariant observable.
(Our use of ``irreducible'' is therefore semantic/universal---``no further compression without loss of invariant content''---and is distinct from,
e.g., irreducibility of a representation in the sense of having no nontrivial invariant subspaces.)

\subsection{Symbolic regression and learned long-distance objects}
Parsimony is a central objective in \emph{symbolic regression} and equation discovery, where one searches for concise analytic expressions
that fit data while controlling complexity.
Classical and modern approaches often treat accuracy and expression complexity as competing objectives and explore a Pareto frontier
\cite{SmitsKotanchek2005Pareto,SchmidtLipson2009Distilling,LaCavaEtAl2021SRBench},
or enforce sparsity/compactness via regularization ideas \cite{Brunton2016SINDy}.
Physics-inspired methods exploit structural properties such as symmetry, separability, and compositionality to find compact expressions
\cite{UdrescuTegmark2020AIFeynman},
and scalable toolchains implement multiobjective equation search for scientific datasets \cite{Cranmer2023PySR};
see also the surveys \cite{MakkeChawla2024SRReview}.

The present paper does \emph{not} propose a new symbolic-regression algorithm.
Instead, it supplies a structural ``pre-parsimony'' step dictated by known physics:
before any attempt to compress long-distance information (whether by a human ansatz, a fit parameterization, or a learned surrogate model),
one should quotient out the internal factorization-scheme redundancy to expose the scheme-invariant interface that actually composes with
short-distance data.
This viewpoint is compatible with treating long-distance correlators as abstract numerical objects (including neural-network surrogates),
because the theorem only assumes the existence of the $A$-module actions that encode admissible scheme changes.
We view this as a foundation for future work in which learned long-distance objects are manipulated in a symbolic calculus that respects
scheme invariance by construction.

\subsection{Aim, contributions, and outline}
We do \emph{not} attempt to prove factorization from first principles.
Rather, we assume as physics input the existence of a factorized/OPE description for a specified observable class in a specified regime $R$
to a specified accuracy $\alpha$.
Within that fixed framework, we:
(i) construct the interface algebra $A$ and module structures induced by standard physics constraints;
(ii) prove the Core Representation Theorem (Theorem~\ref{thm:core}), identifying $C\otimes_A f$ as the universal scheme-invariant carrier;
and (iii) derive a minimal closure principle for completing a generating set of long-distance operators/correlators to an $A$-stable sector.

Section~\ref{sec:physics} reviews collinear factorization and OPE mixing so as to isolate the scheme-change redundancy.
Section~\ref{sec:cat} fixes the categorical environment and recalls balanced morphisms and relative tensor products, with physicist-oriented intuition.
Section~\ref{sec:presentations} formalizes factorization presentations and the scheme-invariant semantic target.
Section~\ref{sec:core-thm} states and proves the Core Representation Theorem and its minimality corollaries,
and explains what the theorem does \emph{not} require of the recomposition operation.
Section~\ref{sec:from-physics} explains how symmetries, locality, and truncation canonically induce the interface algebra and module structures.
Section~\ref{sec:dis-instantiation} instantiates the framework for inclusive DIS, including a concrete toy computation.
Section~\ref{sec:filtered} discusses systematic refinement beyond leading power via filtrations.
Section~\ref{sec:context} gives a brief contextual comparison (SCET/resummation, Hopf-algebraic renormalization, factorization algebras)
and clarifies limitations.

\section{Physics input: collinear factorization and OPE mixing}\label{sec:physics}

\subsection{Collinear factorization in inclusive DIS (leading power)}
In inclusive DIS, the structure functions $F_i(x,Q^2)$ admit, in the Bjorken regime, a leading-power
factorization formula
\begin{equation}\label{eq:dis-factorization}
F_i(x,Q^2)=\sum_{a\in\{q,\bar q,g\}}\bigl(C_i^a(\cdot;Q,\mu)\mconv f_a(\cdot;\mu)\bigr)(x)+\mathcal{O}\!\left(\frac{M^2}{Q^2}\right),
\end{equation}
with Mellin convolution
\begin{equation}\label{eq:mellin-conv}
(C\mconv f)(x)=\int_x^1 \frac{\dd y}{y}\,C(y)\,f\!\left(\frac{x}{y}\right).
\end{equation}
Here $\mu$ denotes the factorization scale (and, when needed, also the renormalization scale; we suppress a separate $\mu_R$)
and the power corrections are controlled by the operator twist expansion.
See \cite{PDG2024StructureFunctions,Collins:FoPQCD,CSS:Factorization}.

\subsection{Finite scheme transformations}
The factorized constituents in \eqref{eq:dis-factorization} are not unique:
a change of factorization scheme (equivalently, a finite change of renormalized operator basis)
acts by an invertible matrix of kernels $Z$ (in flavor space) via
\begin{equation}\label{eq:scheme}
\vf' = Z\mconv \vf,\qquad \vC_i'=\vC_i\mconv Z^{-1},
\end{equation}
leaving the physical composite invariant (to the stated perturbative/power accuracy):
\begin{equation}\label{eq:scheme-inv}
\vC_i\mconv \vf=\vC_i'\mconv \vf'.
\end{equation}
This scheme redundancy is emphasized in standard references \cite{Collins:FoPQCD,PDG2024StructureFunctions} and has also been exploited directly in the DIS literature to construct factorization-scheme-invariant (“physical”) evolution kernels for selected structure-function bases and scheme-invariant non-singlet evolution \cite{BlumleinRavindranvanNeerven2000,BlumleinSaragnese2021}.

\subsection{OPE viewpoint and operator mixing}
In Mellin moment space, the OPE expresses products of currents at short distances in a basis of local operators with
Wilson coefficients \cite{Wilson1969,Wilson1970,WilsonZimmermann1972}. For DIS moments one obtains schematically
\begin{equation}\label{eq:ope-moment}
\int_0^1 \dd x\, x^{n-2}F_i(x,Q^2)\simeq \sum_j C_{i,j}^{(n)}\!\left(\frac{Q^2}{\mu^2},\alpha_s(\mu)\right)\,
\langle P|O_j^{(n)}(\mu)|P\rangle+\mathcal{O}\!\left(\frac{1}{Q^{p}}\right),
\end{equation}
where $O_j^{(n)}$ are twist-two operators (higher twists contribute to power-suppressed terms).
Renormalization induces operator mixing
\begin{equation}\label{eq:mix}
O_{j,\mathrm{bare}}^{(n)}=\sum_k Z_{jk}^{(n)}(\mu)\,O_k^{(n)}(\mu),
\end{equation}
and Wilson coefficients transform with the inverse so that physical moments are $\mu$-independent up to higher-twist terms:
\begin{equation}\label{eq:coeff-inv}
C_{i,\mathrm{bare}}^{(n)}=\sum_k C_{i,k}^{(n)}(\mu)\,\bigl(Z^{(n)}(\mu)\bigr)^{-1}_{kj}.
\end{equation}
This is the moment-space analogue of \eqref{eq:scheme}--\eqref{eq:scheme-inv}.

\subsection{Structural takeaway}
Both the $x$-space and moment-space viewpoints share the same algebraic pattern:
\begin{quote}
\emph{There is a class of admissible finite redefinitions (scheme transformations) acting on the factors such
that the recomposed observable is invariant.}
\end{quote}
Our categorical formalism encodes the admissible redefinitions as an algebra action and identifies the universal
scheme-invariant composite.

\section{Categorical framework}\label{sec:cat}

\subsection{Physics intuition: quotienting scheme redundancy}
For a physicist, the relative tensor product $C\otimes_A f$ should be read as a \emph{quotient construction}:
it is obtained from the naive composite $C\otimes f$ by imposing the balancing relation
\[
(C\cdot a)\otimes f \;\sim\; C\otimes (a\cdot f)\qquad (a\in A),
\]
i.e.\ by declaring that ``inserting an admissible finite counterterm $a$ on the coefficient side'' is equivalent to
``inserting it on the correlator side.'' This is exactly the algebraic content of scheme invariance
(e.g.\ $C\mconv Z^{-1}$ paired with $Z\mconv f$ gives the same observable).
In this sense, $C\otimes_A f$ plays the same conceptual role as taking coinvariants under a gauge action:
one modds out internal redundancies but does not discard any invariant information.
Formally, the ``modding out'' is implemented by a coequalizer, which in linear settings is simply a quotient by a
subspace generated by the balancing relations.

\subsection{Ambient symmetric monoidal category}
We work in a symmetric monoidal category $(\cV,\otimes,\mathbf{1})$ modeling the chosen recomposition calculus.
For collinear DIS in $x$-space, $\otimes$ may be realized by Mellin convolution \eqref{eq:mellin-conv};
in Mellin moment space it is realized by pointwise multiplication.
The theorem itself requires only structural properties.

\begin{assumption}[Existence of balanced tensor products]\label{ass:ambient}
Let $(\cV,\otimes,\mathbf{1})$ be a cocomplete symmetric monoidal category such that $\otimes$ preserves the coequalizers
used to define relative tensor products, separately in each variable.
\end{assumption}

\begin{remark}[Why coequalizers are not exotic in QCD practice]\label{rem:coeq}
If $\cV=\mathrm{Vect}_k$ (or $\cV=\mathrm{Mod}_R$), then all coequalizers exist and are computed as quotients.
Concretely, for a right $A$-module $N$ and a left $A$-module $M$,
\[
N\otimes_A M \;\cong\; (N\otimes_k M)\Big/\mathrm{span}_k\Big\{(n\cdot a)\otimes m - n\otimes (a\cdot m)
:\; n\in N,\; m\in M,\; a\in A\Big\}.
\]
(For $\mathrm{Mod}_R$, replace $\otimes_k$ and $\mathrm{span}_k$ by $\otimes_R$ and the $R$-submodule generated.)
This is the algebraic shadow of the familiar statement that scheme choices change $(C,f)$ but not $C\mconv f$.
Analytic realizations relevant to QCD (kernels as distributions with plus-prescriptions, end-point singularities, etc.)
can be accommodated by choosing $\cV$ to be a cocomplete linear tensor category of generalized functions
(e.g.\ suitable topological vector spaces) in which convolution/multiplication is associative and continuous; we keep these analytic requirements explicit in Assumption~\ref{ass:ambient}
and the required quotient/coequalizer exists; Assumption~\ref{ass:ambient} packages these analytic hypotheses.
For concrete distributional convolution calculi used in QCD factorization, see e.g.\ \cite{Collins:FoPQCD}.
\end{remark}

\subsection{Algebra objects and modules}
\begin{definition}[Algebra object]
An \emph{algebra object} (monoid object) in $\cV$ is an object $A$ equipped with multiplication
$m:A\otimes A\to A$ and unit $\eta:\mathbf{1}\to A$ satisfying associativity and unitality.
\end{definition}

\begin{definition}[Modules]
A \emph{left} $A$-module is an object $M$ with an action $\lambda:A\otimes M\to M$; a \emph{right} $A$-module
is an object $N$ with an action $\rho:N\otimes A\to N$.
\end{definition}

\subsection{Balanced morphisms and relative tensor products}
\label{subsec:balanced-and-relative-tensor}

\paragraph{Coherence convention.}
Throughout, $\mathcal V$ is a (symmetric) monoidal category with associator
$\alpha_{X,Y,Z} : (X\otimes Y)\otimes Z \xrightarrow{\sim} X\otimes (Y\otimes Z)$
and unitors. In physics notation one often suppresses parentheses and coherence maps.
Since the balancing relation mixes the two parenthesizations, we write the relevant associator
explicitly below; this is the fully correct non-strict formulation and is the one formalized in Lean.

\begin{remark}[Coherence shorthand]\label{rem:coherence-shorthand}
For the remainder of the paper we freely identify the two parenthesizations $((X\otimes Y)\otimes Z)$ and
$(X\otimes (Y\otimes Z))$ via the associator, and we write $X\otimes Y\otimes Z$ for either parenthesization
when the choice is immaterial by coherence. In particular, expressions such as $\mathrm{id}_X\otimes \lambda$
are understood to include the necessary rebracketing isomorphisms.
\end{remark}

\begin{definition}[{$A$-balanced morphism}]
\label{def:balanced-morphism}
Let $A$ be an algebra object in $\mathcal V$. Let $(N,\rho_N)$ be a right $A$-module and
$(M,\lambda_M)$ a left $A$-module, with actions
\[
\rho_N : N\otimes A \to N,
\qquad
\lambda_M : A\otimes M \to M.
\]
(We write $\rho_N,\lambda_M$ to avoid confusion with the \emph{unitors} of the monoidal structure.)

For any object $X\in\mathcal V$ and any morphism $\varphi : N\otimes M \to X$, define the two
parallel action-induced morphisms
\[
r := (\rho_N \otimes \mathrm{id}_M) : (N\otimes A)\otimes M \to N\otimes M,
\]
\[
\ell := \alpha_{N,A,M}\,;\,(\mathrm{id}_N \otimes \lambda_M) : (N\otimes A)\otimes M \to N\otimes M.
\]
We say that $\varphi$ is \emph{$A$-balanced} if the following equality holds:
\begin{equation}
\label{eq:balanced-eq}
r\,;\,\varphi \;=\; \ell\,;\,\varphi
\qquad\text{as morphisms } (N\otimes A)\otimes M \to X.
\end{equation}
Equivalently, the diagram
\begin{equation}
\label{eq:balanced-diagram}
\begin{tikzcd}[column sep=large,row sep=large]
 (N\otimes A)\otimes M \ar[r,"r"] \ar[d,"\ell"'] & N\otimes M \ar[d,"\varphi"] \\
 N\otimes M \ar[r,"\varphi"'] & X
\end{tikzcd}
\end{equation}
commutes.
\end{definition}

\begin{definition}[{Relative tensor product}]
\label{def:relative-tensor-product}
Assume the coequalizer of the parallel pair $(r,\ell)$ exists in $\mathcal V$.
The \emph{relative tensor product} (or \emph{balanced tensor product})
$N\otimes_A M$ is defined to be the coequalizer
\begin{equation}
\label{eq:coeq-def}
 (N\otimes A)\otimes M \;\substack{\xrightarrow{\ r\ }\ \\[-1mm] \xrightarrow[\ \ell\ ]{}}\;
 N\otimes M \xrightarrow{q} N\otimes_A M,
\end{equation}
so that $r\,;\,q = \ell\,;\,q$ and $q$ is universal with this property.
\end{definition}

\begin{proposition}[{Universal property of $N\otimes_A M$}]
\label{prop:universal-property-relative-tensor}
Under the hypotheses above, for every $X\in\mathcal V$ composition with $q$ induces a natural bijection
\begin{equation}
\label{eq:universal-property-bijection}
\Hom_{\mathcal V}(N\otimes_A M,\,X)
\;\cong\;
\{\varphi\in\Hom_{\mathcal V}(N\otimes M,\,X)\mid \varphi \text{ is $A$-balanced}\}.
\end{equation}
Equivalently: a morphism $\varphi:N\otimes M\to X$ is $A$-balanced if and only if there exists a unique
$\overline{\varphi}:N\otimes_A M\to X$ with $\varphi = q\,;\,\overline{\varphi}$.
\end{proposition}

\begin{proof}
This is exactly the universal property of the coequalizer \eqref{eq:coeq-def}. Concretely:
given $\overline{\varphi}:N\otimes_A M\to X$, the composite $q\,;\,\overline{\varphi}$ is balanced because
$r\,;\,q=\ell\,;\,q$. Conversely, if $\varphi:N\otimes M\to X$ is balanced, i.e. satisfies
\eqref{eq:balanced-eq}, then by the coequalizer universal property there exists a unique
$\overline{\varphi}$ with $\varphi=q\,;\,\overline{\varphi}$. Naturality in $X$ is immediate.
\end{proof}

\section{Factorization presentations and scheme-invariant semantics}\label{sec:presentations}

The Core Representation Theorem is a representability statement about \emph{scheme-invariant} (balanced) pairings.
To emphasize how this interfaces with physics practice, we fix an explicit semantic target (``agreement to a stated accuracy'')
and then define what we mean by a factorization \emph{presentation} of that semantic target.

\subsection{Accuracy data and semantic target}
Fix a kinematic regime $R$ (e.g.\ the Bjorken regime) and an \emph{accuracy datum} $\alpha$ (e.g.\ truncation in twist and perturbative order).
We encode ``agreement up to $\alpha$'' by an object $O_{R,\alpha}\in \cV$ representing observables on $R$ modulo terms beyond $\alpha$.
In linear settings, one may take $O_{R,\alpha}$ to be a quotient $O_R / \mathcal{R}_{>\alpha}$ by a remainder subobject.

\begin{remark}
This semantic step is where ``no information loss'' is enforced: the comparison between presentations is performed only after modding out
the remainder class $\mathcal{R}_{>\alpha}$. If one strengthens $\alpha$ (e.g.\ includes higher twist), one refines the semantic target.
Section~\ref{sec:filtered} shows how to organize these refinements systematically.
\end{remark}

\subsection{Presentation category and scheme groupoid}
Fix an algebra object $A$ in $\cV$ (the interface algebra) and the semantic target $O_{R,\alpha}$.

\begin{definition}[Presentation]\label{def:presentation}
A \emph{factorization presentation} (at fixed $(R,\alpha)$) is a triple $(N,M,\phi)$ where $N$ is a right $A$-module,
$M$ is a left $A$-module, and $\phi:N\otimes M\to O_{R,\alpha}$ is an $A$-balanced morphism.
\end{definition}

\begin{definition}[Presentation category]\label{def:P}
Let $\mathsf{P}(A;O_{R,\alpha})$ be the category whose objects are presentations $(N,M,\phi)$ as in Definition~\ref{def:presentation}.
A morphism $(N,M,\phi)\to (N',M',\phi')$ is a pair of $A$-linear maps $u:N\to N'$ (right-linear) and $v:M\to M'$ (left-linear)
such that $\phi' \circ (u\otimes v)=\phi$.
Equivalently, morphisms are exactly those pairs $(u,v)$ for which the following triangle commutes:
\begin{equation}\label{eq:presentation-morphism}
\begin{tikzcd}
N\otimes M \ar[rr, "\phi"] \ar[dr, swap, "u\otimes v"] && O_{R,\alpha}\\
& N'\otimes M' \ar[ur, swap, "\phi'"] &
\end{tikzcd}
\end{equation}
\end{definition}

\begin{remark}[Scheme groupoid and its QCD reading]\label{rem:scheme-groupoid}
In many physics applications one restricts attention to \emph{invertible} scheme changes, and one passes to the maximal subgroupoid
$\mathsf{S}(A;O_{R,\alpha})\subset \mathsf{P}(A;O_{R,\alpha})$ consisting of isomorphisms.

For inclusive DIS, one takes $A=A_{\mathrm{coll}}$ (Section~\ref{sec:dis-instantiation}). A typical object is the familiar pair
$(\vC_i,\vf)$ together with the convolution evaluation $\Phi_i(\vC_i,\vf)=\sum_a C_i^a\mconv f_a$.
An arrow in $\mathsf{S}(A_{\mathrm{coll}};O_{R,\alpha})$ is implemented by an invertible kernel matrix $Z\in A_{\mathrm{coll}}^\times$:
it acts on PDFs by $v_Z:\vf\mapsto Z\cdot \vf$ and on coefficients by $u_Z:\vC_i\mapsto \vC_i\cdot Z^{-1}$, and the defining commutativity
$\Phi_i\circ(u_Z\otimes v_Z)=\Phi_i$ is exactly the standard scheme-invariance identity
$(\vC_i\mconv Z^{-1})\mconv (Z\mconv \vf)=\vC_i\mconv \vf$.

Symmetry constraints (e.g.\ flavor $SU(n_f)$ in massless QCD, charge conjugation, and the twist/spin quantum numbers fixed by $(R,\alpha)$)
restrict which $Z$ are admissible: only kernels preserving the corresponding block decomposition of operator/PDF sectors appear.
\end{remark}

\subsection{Balanced semantics as a functor}
For fixed right/left $A$-modules $(N,M)$, define the functor
\begin{equation}\label{eq:Bal}
\Bal_A(N,M):\cV\to \mathbf{Set},\qquad
X\mapsto \{\varphi:N\otimes M\to X \mid \varphi \text{ is $A$-balanced}\}.
\end{equation}
The physical meaning is:
\begin{quote}
\emph{Scheme-invariant observables are exactly the $A$-balanced evaluations.}
\end{quote}
The Core Representation Theorem identifies the representing object for \eqref{eq:Bal}.

\section{Core Representation Theorem}\label{sec:core-thm}

\begin{theorem}[Core Representation Theorem]
\label{thm:core}
Assume the relative tensor product $N\otimes_A M$ exists as in Definition~\ref{def:relative-tensor-product}.
Then the functor
\[
\Bal_A(N,M) : \mathcal V \to \mathbf{Set},
\qquad
X \mapsto \{\varphi:N\otimes M\to X \mid \varphi\ \text{$A$-balanced}\},
\]
is representable by $N\otimes_A M$. Equivalently, for every $X\in\mathcal V$ there is a natural bijection
\begin{equation}
\label{eq:core-bijection}
\Hom_{\mathcal V}(N\otimes_A M, X)\;\cong\;\Bal_A(N,M)(X),
\end{equation}
carried by $\overline{\varphi}\mapsto q\,;\,\overline{\varphi}$.

Moreover:
\begin{enumerate}
\item[(i)] \textbf{(Universal factorization)} Every $A$-balanced morphism $\varphi:N\otimes M\to X$
factors uniquely as
\[
N\otimes M \xrightarrow{q} N\otimes_A M \xrightarrow{\overline{\varphi}} X.
\]

\item[(ii)] \textbf{(Terminal scheme-invariant quotient)} Let $\mathsf Q$ be the category whose objects are
epimorphisms $p:N\otimes M \twoheadrightarrow Q$ in $\mathcal V$ such that for every $X\in\mathcal V$,
every $A$-balanced morphism $\varphi:N\otimes M\to X$ factors through $p$ (necessarily uniquely, since $p$
is epi). A morphism $(p:N\otimes M\twoheadrightarrow Q)\to (p':N\otimes M\twoheadrightarrow Q')$
is a morphism $u:Q\to Q'$ with $p' = p\,;\,u$.
Then $q:N\otimes M \twoheadrightarrow N\otimes_A M$ is a terminal object of $\mathsf Q$.
\end{enumerate}
\end{theorem}

\begin{proof}
The representability statement \eqref{eq:core-bijection} and part (i) are exactly
Proposition~\ref{prop:universal-property-relative-tensor}.

For (ii), first note that $q$ is an epimorphism (coequalizer maps are epis) and that every
balanced $\varphi$ factors through $q$ by (i), hence $q$ is an object of $\mathsf Q$.

Now let $p:N\otimes M \twoheadrightarrow Q$ be any object of $\mathsf Q$.
The coequalizer map $q:N\otimes M\to N\otimes_A M$ is itself $A$-balanced, because
$r\,;\,q=\ell\,;\,q$ by Definition~\ref{def:relative-tensor-product}. By the defining property of objects in
$\mathsf Q$, this balanced morphism $q$ must factor through $p$, i.e. there exists a morphism
$u:Q\to N\otimes_A M$ such that
\[
q = p\,;\,u.
\]
This precisely exhibits a morphism $p\to q$ in $\mathsf Q$. Uniqueness of $u$ follows because $p$ is epi.
Therefore $q$ is terminal in $\mathsf Q$.
\end{proof}

\subsection{Immediate corollaries: invariance and minimality}
\begin{corollary}[Invariance under module isomorphisms]\label{cor:inv}
If $u:N\to N'$ is an isomorphism of right $A$-modules and $v:M\to M'$ is an isomorphism of left $A$-modules,
then there is a canonical isomorphism
\[
N\otimes_A M \cong N'\otimes_A M'.
\]
\end{corollary}

\begin{proof}
Functoriality of relative tensor products in bimodule maps yields the claimed isomorphism.
\end{proof}

\begin{corollary}[Minimality for monotone complexity measures]\label{cor:min}
Let $d:\mathrm{Ob}(\cV)\to (P,\le)$ be any functional monotone under quotients (epimorphisms).
Then $d(N\otimes_A M)\le d(N\otimes M)$. Moreover, among quotients in $\mathsf{Q}$ of Theorem~\ref{thm:core}(ii),
$N\otimes_A M$ is minimal in this sense.
\end{corollary}

\begin{proof}
The map $q$ exhibits $N\otimes_A M$ as a quotient of $N\otimes M$, hence monotonicity gives the inequality.
Terminality in $\mathsf{Q}$ implies minimality among scheme-invariant quotients.
\end{proof}

\subsection{Recomposition}\label{sec:op-agnostic}
The Core Representation Theorem is \emph{not} a statement about any particular analytic formula (a specific convolution, integral transform, or closed form).
It uses only the following structural inputs:
\begin{enumerate}[leftmargin=2.0em]
\item a symmetric monoidal product $\otimes$ implementing ``recomposition'' of the two factors (associative up to coherent isomorphism),
\item an interface algebra object $A$ acting on the two factors as a right and a left module, and
\item existence of the coequalizer defining the relative tensor product $N\otimes_A M$ (Assumption~\ref{ass:ambient}).
\end{enumerate}
Accordingly, the theorem \emph{cannot} be applied precisely when these structural hypotheses fail.
In QCD language, this includes situations where the would-be recomposition rule is genuinely non-compositional (non-associative or lacking the needed interchange/Fubini properties),
where finite counterterms cannot be represented as a legitimate algebra action on \emph{both} factors (often signaling missing leading regions or incomplete subtraction data),
or where the chosen analytic category does not admit the required quotients/coequalizers for the kernel objects under consideration.

\subsection{Changes of representation}
It is common in perturbative QCD to pass to a representation (Mellin moments, Laplace, Fourier, \dots) in which the recomposition law becomes
simpler (e.g.\ multiplication rather than convolution). Categorically, such a change is modeled by a strong monoidal functor.

\begin{proposition}[Core commutes with strong monoidal transforms]\label{prop:monoidal}
Let $T:(\cV,\otimes,\mathbf{1})\to (\cV',\otimes',\mathbf{1}')$ be a strong monoidal functor preserving the coequalizers used to define relative tensor products.
Then, naturally in right/left modules $(N,M)$ over $A$,
\[
T(N\otimes_A M)\cong T(N)\otimes_{T(A)} T(M),
\]
and $T$ carries $A$-balanced maps to $T(A)$-balanced maps.
\end{proposition}

\begin{proof}
Write the right/left actions as $\rho:N\otimes A\to N$ and $\lambda:A\otimes M\to M$.
Applying $T$ to these and using the strong monoidal structure gives induced actions
\[
T(N)\otimes' T(A)\xrightarrow{\ \cong\ }T(N\otimes A)\xrightarrow{T(\rho)}T(N),\qquad
T(A)\otimes' T(M)\xrightarrow{\ \cong\ }T(A\otimes M)\xrightarrow{T(\lambda)}T(M),
\]
so $T(N)$ is a right $T(A)$-module and $T(M)$ is a left $T(A)$-module.
Now apply $T$ to the coequalizer diagram \eqref{eq:coeq-def} defining $N\otimes_A M$.
Let $q':T(N)\otimes' T(M)\twoheadrightarrow T(N)\otimes_{T(A)}T(M)$ denote the coequalizer defining the relative tensor product in $\cV'$, and
let $\psi_{N,M}:T(N\otimes M)\xrightarrow{\ \cong\ }T(N)\otimes' T(M)$ be the strong monoidal coherence isomorphism.
By hypothesis, $T$ preserves the relevant coequalizer, so $T(q):T(N\otimes M)\to T(N\otimes_A M)$ is (up to canonical identifications)
a coequalizer of the transported parallel pair. Hence there is a unique induced isomorphism
$\theta:T(N\otimes_A M)\xrightarrow{\ \cong\ }T(N)\otimes_{T(A)}T(M)$ making the comparison diagram commute:
\begin{equation}\label{eq:monoidal-compare}
\begin{tikzcd}[column sep=huge]
T(N\otimes M) \ar[r, "T(q)"] \ar[d, swap, "\psi_{N,M}"]
  & T(N\otimes_A M) \ar[d, dashed, "\theta"]\\
T(N)\otimes' T(M) \ar[r, swap, "q'"]
  & T(N)\otimes_{T(A)}T(M)
\end{tikzcd}
\end{equation}
This yields the claimed natural isomorphism.
Finally, if $\phi:N\otimes M\to X$ is $A$-balanced, then \eqref{eq:balanced-eq} is an equality of morphisms
$(N\otimes A)\otimes M\to X$. Applying $T$ gives equality of the corresponding morphisms out of
$T(N)\otimes' T(A)\otimes' T(M)$, i.e.\ $T(\phi)$ is $T(A)$-balanced.
\end{proof}

\section{Interface algebra}\label{sec:from-physics}

The categorical theorem becomes a physics tool once the interface algebra and module structures are forced by standard QCD input:
locality/OPE, symmetry constraints, truncation data, and renormalization mixing.
We make this dependence explicit and isolate the resulting mathematical object.

\subsection{Operator sector selection}
Locality and the OPE provide a short-distance expansion in a basis of local operators organized by twist
\cite{Wilson1969,Wilson1970,WilsonZimmermann1972}. Fix an accuracy datum $\alpha$ (twist/power truncation and perturbative order).
Let $\mathsf{Op}_\alpha$ denote the truncated span of gauge-invariant local operators relevant to the observable class, modulo standard redundancies
(e.g.\ total derivatives as appropriate to the semantics).

\subsection{Canonical interface algebra}
Let $G$ denote the symmetry acting on $\mathsf{Op}_\alpha$ that we insist all admissible finite
counterterms respect (e.g.\ gauge invariance and the residual spacetime/flavor symmetries
appropriate to the kinematic regime and truncation). Assume we are in a semisimple setting so
that $\mathsf{Op}_\alpha$ is completely reducible as a $G$-representation. Let $\Lambda_\alpha$ denote
the (finite) set of isomorphism classes of irreducible $G$-representations that occur in
$\mathsf{Op}_\alpha$. For each $\lambda\in\Lambda_\alpha$, choose a representative irreducible
$G$-module $V_\lambda$ and define the corresponding multiplicity space
\[
M_\lambda \;:=\; \mathrm{Hom}_G\!\big(V_\lambda,\mathsf{Op}_\alpha\big),
\]
so that $\dim M_\lambda$ is the multiplicity of $V_\lambda$ inside $\mathsf{Op}_\alpha$ (and $M_\lambda$
carries the trivial $G$-action). Then one has the $G$-isotypic decomposition
\begin{equation}\label{eq:isotypic}
\mathsf{Op}_\alpha \;\cong\; \bigoplus_{\lambda\in\Lambda_\alpha} V_\lambda \otimes M_\lambda,
\end{equation}
where $G$ acts on each $V_\lambda$ and trivially on each $M_\lambda$.

In this decomposition, any $G$-equivariant endomorphism $T\in \mathrm{End}_G(\mathsf{Op}_\alpha)$
cannot mix inequivalent types $\lambda\neq\lambda'$. Moreover, by Schur's lemma (in particular,
when $\mathrm{End}_G(V_\lambda)\cong k$), the restriction of $T$ to each summand is of the form
\[
T|_{V_\lambda\otimes M_\lambda} \;=\; \mathrm{id}_{V_\lambda}\otimes T_\lambda
\qquad\text{for a unique } T_\lambda\in \mathrm{End}(M_\lambda).
\]
This is the sense in which $G$-equivariant endomorphisms act \emph{trivially} on the irrep factor
$V_\lambda$ and \emph{freely} on the multiplicity factor $M_\lambda$: within each isotypic block,
the $G$-equivariance condition imposes no further constraint on $T_\lambda$ beyond linearity.
Equivalently, the commutant (interface) algebra is
\begin{equation}\label{eq:commutant}
\mathrm{End}_G(\mathsf{Op}_\alpha)\;\cong\;\bigoplus_{\lambda\in\Lambda_\alpha}\mathrm{End}(M_\lambda).
\end{equation}
This makes precise a familiar physics statement: symmetry forbids mixing between inequivalent
operator quantum numbers, so admissible finite counterterms are block-diagonal.
(For a standard mathematical reference, see e.g.\ \cite{FultonHarris}.)


\begin{definition}[Interface algebra at accuracy $\alpha$]\label{def:Aalpha}
Let $A_\alpha$ be the algebra object whose elements are admissible finite renormalization/subtraction maps acting within $\mathsf{Op}_\alpha$,
compatible with the chosen symmetry constraints. Concretely, in semisimple symmetry settings one may take $A_\alpha\simeq \mathrm{End}_G(\mathsf{Op}_\alpha)$,
as in \eqref{eq:commutant}. We call $A_\alpha$ the \emph{interface algebra} at accuracy $\alpha$.
\end{definition}

\begin{remark}[What is and is not encoded in $A_\alpha$]
$A_\alpha$ encodes the \emph{finite} freedom in defining renormalized operators (and hence collinear subtraction conventions)
within the truncated sector. It does \emph{not} encode the hard physics of establishing that a given observable admits a factorized description in $R$,
nor does it encode the power-suppressed remainder beyond $\alpha$; those are placed in the semantic target $O_{R,\alpha}$.
\end{remark}

\subsection{Coefficient and hadronic data as modules}
Wilson coefficients form a right $A_\alpha$-module (they transform by inverse finite renormalizations),
while hadronic matrix elements/parton correlators form a left $A_\alpha$-module (they transform covariantly):
\[
f\mapsto a\cdot f,\qquad C\mapsto C\cdot a^{-1},\qquad a\in A_\alpha^\times.
\]
This is exactly the structural content of \eqref{eq:scheme} and \eqref{eq:coeff-inv}.

\subsection{Universal observable map}
Given module data $(C,f)$ and an $A_\alpha$-balanced evaluation $\Phi:C\otimes f\to O_{R,\alpha}$,
Theorem~\ref{thm:core} produces a unique induced map
\begin{equation}\label{eq:coremap}
\overline{\Phi}:\ C\otimes_{A_\alpha} f \longrightarrow O_{R,\alpha}.
\end{equation}
This is the precise sense in which the core $C\otimes_{A_\alpha}f$ is a \emph{representation-theoretic lift} of the physical observable:
it is the terminal object through which all balanced (scheme-invariant) presentations factor.

Because $\Phi$ is $A_\alpha$-balanced, it coequalizes the two action maps and hence admits the canonical lift
$\overline{\Phi}$ making the following diagram commute:
\begin{equation}\label{eq:core-lift-diagram}
\begin{tikzcd}[column sep=huge, row sep=large]
C\otimes A_\alpha\otimes f
  \ar[r, shift left=0.6ex, "\rho\otimes\mathrm{id}"]
  \ar[r, shift right=0.6ex, swap, "\mathrm{id}\otimes\lambda"]
&
C\otimes f \ar[r, "q"] \ar[dr, "\Phi"']
&
C\otimes_{A_\alpha} f \ar[d, "\overline{\Phi}"] \\
&& O_{R,\alpha}
\end{tikzcd}
\end{equation}

For orientation, one may summarize the induced objects/data as the following ``physics-to-core'' chain:
\begin{equation}\label{eq:physics-to-core}
\begin{tikzcd}[column sep=large]
\text{input }(G,\alpha) \ar[r, "{\text{\scriptsize OPE}}"]
  & \mathsf{Op}_\alpha \ar[r, "{\text{\scriptsize commut.}}"]
  & A_\alpha \ar[r, "{\text{\scriptsize module}}"]
  & (C,f,\Phi) \ar[r, "{\text{\scriptsize coeq.}}"]
  & C\otimes_{A_\alpha} f \ar[r, "\overline{\Phi}"]
  & O_{R,\alpha},
\end{tikzcd}
\end{equation}
which summarizes a concrete construction pipeline.
The guiding principle is that \emph{physical content} is what survives admissible scheme redefinitions,
and the categorical core $C\otimes_{A_\alpha} f$ is the universal carrier of precisely that invariant content.

\smallskip
\noindent\textbf{(i) Input $(G,\alpha)$}:
\begin{itemize}
\item a symmetry specification $G$ (in practice: gauge invariance and the additional global/discrete symmetries
relevant to the observable class, such as flavor symmetry blocks, charge conjugation, parity, and the Lorentz/twist grading), and
\item an accuracy datum $\alpha$ (e.g.\ a twist/power truncation and a perturbative order truncation),
together with an implicit choice of kinematic regime/phase-space domain $R$ defining the semantic target $O_{R,\alpha}$.
\end{itemize}
The role of $(G,\alpha)$ is to fix \emph{what is admissible}: which operators belong to the truncation,
which finite counterterms are allowed, and which equivalences should be regarded as ``pure scheme.''

\smallskip
\noindent\textbf{(ii) $(G,\alpha)\xrightarrow{\mathrm{OPE}} \mathrm{Op}_\alpha$:}
Locality and short-distance dominance supply an OPE description in which the observable is controlled,
to the stated accuracy, by a \emph{finite/truncated sector} of renormalized composite operators.
We denote by $\mathrm{Op}_\alpha$ the resulting operator sector: the span (or appropriate completion)
of those gauge-invariant local operators that contribute in the regime $R$ through the truncation $\alpha$,
modulo standard redundancies (total derivatives, equations of motion as appropriate for the observable, etc.).
This is the domain on which renormalization mixing acts and from which long-distance hadronic data is extracted.

\smallskip
\noindent\textbf{(iii) $\mathrm{Op}_\alpha\xrightarrow{\mathrm{commut.}} A_\alpha$:}
Renormalization and collinear subtraction induce \emph{operator mixing} within $\mathrm{Op}_\alpha$,
and finite scheme changes correspond to \emph{finite, invertible} redefinitions within this same sector.
The interface algebra $A_\alpha$ packages these admissible finite redefinitions into an algebra object.
A useful abstract characterization is:
$A_\alpha \;\subseteq\; \End(\mathrm{Op}_\alpha)$
 as the algebra of admissible endomorphisms compatible with the symmetry data $G$ and truncation $\alpha$.
The label ``commut.'' is meant in the standard sense of \emph{commutant / symmetry-preserving endomorphisms}:
$A_\alpha$ is obtained by restricting $\End(\mathrm{Op}_\alpha)$ to those maps that preserve the
$G$-decomposition (flavor blocks, twist grading, Lorentz structure, \emph{etc.}) and hence represent legitimate
finite counterterms/scheme transformations at the stated accuracy.
In moment space this is literally a matrix algebra of finite mixing maps on each spin/twist block;
in $x$-space it is realized by matrices of distribution-valued kernels with multiplication given by
(matrix) Mellin convolution.

\smallskip
\noindent\textbf{(iv) $A_\alpha\xrightarrow{\mathrm{module}} (C,f,\Phi)$:}
Once $A_\alpha$ is fixed, the standard physics covariance of scheme changes forces a bimodule pattern:
\begin{itemize}
\item the short-distance coefficient data $C$ carries a \emph{right} $A_\alpha$-module structure
(coefficients transform contravariantly under finite renormalizations), and
\item the long-distance hadronic data $f$ carries a \emph{left} $A_\alpha$-module structure
(matrix elements/PDFs transform covariantly).
\end{itemize}
The ``evaluation'' (or recomposition) map
\[
\Phi:\; C\otimes f \longrightarrow O_{R,\alpha}
\]
is the abstract encoding of the factorized formula (convolution in $x$-space, multiplication in moment space, \emph{etc.}).
The defining physical requirement is \emph{scheme invariance}:
inserting an admissible finite counterterm on the coefficient side or on the correlator side must give the same observable.
Categorically, this is precisely the balancing condition
\[
\Phi\bigl((C\cdot a)\otimes f\bigr) \;=\; \Phi\bigl(C\otimes (a\cdot f)\bigr)\qquad (a\in A_\alpha),
\]
so that $\Phi$ is an $A_\alpha$-balanced morphism.

\smallskip
\noindent\textbf{(v) $(C,f,\Phi)\xrightarrow{\mathrm{coeq.}} C\otimes_{A_\alpha} f$:}
The relative tensor product $C\otimes_{A_\alpha} f$ is defined as the coequalizer imposing the balancing relation.
In additive/linear settings one can think of it concretely as the quotient
\[
C\otimes_{A_\alpha} f
\;\cong\;
(C\otimes f)\Big/\big\langle (C\cdot a)\otimes f - C\otimes(a\cdot f)\;:\;a\in A_\alpha\big\rangle,
\]
i.e.\ the naive composite $C\otimes f$ with \emph{exactly} the scheme redundancy collapsed.
This is the precise sense in which the construction is ``parsimonious'': it removes no information except that which is
provably unphysical (pure scheme), and it removes \emph{all} such redundancy forced by the $A_\alpha$-action.

\smallskip
\noindent\textbf{(vi) $C\otimes_{A_\alpha} f\xrightarrow{\ \overline{\Phi}\ } O_{R,\alpha}$:}
By the universal property of the coequalizer, any $A_\alpha$-balanced evaluation $\Phi$ descends uniquely:
there exists a unique morphism
\[
\overline{\Phi}:\; C\otimes_{A_\alpha} f \longrightarrow O_{R,\alpha}
\quad\text{with}\quad
\Phi = \overline{\Phi}\circ q,
\]
where $q:C\otimes f\to C\otimes_{A_\alpha} f$ is the canonical quotient map.
Thus $C\otimes_{A_\alpha} f$ is the \emph{universal scheme-invariant carrier} of the observable content:
every scheme-invariant evaluation factors through it, and any further quotient would necessarily
identify two presentations that are distinguished by some scheme-invariant observable.
In particular, the ``core'' object does not depend on a preferred scheme; it depends only on
the symmetry/accuracy input $(G,\alpha)$ and the resulting admissible interface action.

\subsection{A minimal closure principle at fixed accuracy}\label{sec:closure}
A persistent practical question is: given a small set of ``primitive'' long-distance operators/correlators selected by
symmetry and power counting, what is the minimal enlargement required to obtain a sector stable under admissible
finite renormalizations (and hence suitable for a scheme-invariant core construction)?

\begin{definition}[\texorpdfstring{$A_\alpha$}{A_alpha}-closure]\label{def:closure}
Let $G_0\subseteq \mathsf{Op}_\alpha$ be a subset (or subobject) generating the desired long-distance sector.
Define its \emph{$A_\alpha$-closure} by
\[
\langle G_0\rangle_{A_\alpha} := A_\alpha\cdot G_0,
\]
the smallest left $A_\alpha$-submodule of $\mathsf{Op}_\alpha$ containing $G_0$.
\end{definition}

\begin{proposition}[Minimal scheme/RG-stable sector]\label{prop:closure}
$\langle G_0\rangle_{A_\alpha}$ is the unique minimal left $A_\alpha$-stable submodule of $\mathsf{Op}_\alpha$ containing $G_0$.
In particular, any long-distance sector containing $G_0$ and closed under admissible finite renormalizations must contain
$\langle G_0\rangle_{A_\alpha}$.
\end{proposition}

\begin{proof}
By construction $\langle G_0\rangle_{A_\alpha}$ is $A_\alpha$-stable and contains $G_0$.
If $M\subseteq \mathsf{Op}_\alpha$ is any $A_\alpha$-stable submodule containing $G_0$, then
$A_\alpha\cdot G_0\subseteq A_\alpha\cdot M\subseteq M$, hence $\langle G_0\rangle_{A_\alpha}\subseteq M$.
\end{proof}

\section{Instantiation: inclusive DIS in \texorpdfstring{$x$}{x}-space and moment space}\label{sec:dis-instantiation}

We now specify $A$, $N$, $M$, and $O$ for collinear DIS, show that the physical factorization pairing is balanced, and therefore descends canonically to the core.

\subsection{A monoidal model for \texorpdfstring{$x$}{x}-space collinear factorization}
Choose $\cV$ to model (matrix-valued) generalized functions on $(0,1)$ with monoidal product $\otimes$ realized as Mellin convolution $\mconv$
and unit $\delta(1-x)$; cf.\ \eqref{eq:mellin-conv}. We keep the analytic details abstract, using only Assumption~\ref{ass:ambient}.

\subsection{The collinear interface algebra from the twist-two operator sector}\label{sec:Acoll}

We now make Definition~\ref{def:Aalpha} completely explicit in the DIS leading-twist setting, by choosing a truncated operator basis
and reading off the corresponding commutant algebra of admissible finite counterterms.

\paragraph{Twist-two operator sector and symmetry blocks.}
At leading power in $\DIS$, the relevant $\OPE$ sector consists of gauge-invariant twist-two operators.
In moment space, for each spin (Mellin moment) $n\ge 2$ one has a finite-dimensional vector space
$\mathsf{Op}^{(n)}$ spanned by the corresponding local twist-two operators modulo the standard redundancies
(e.g.\ total derivatives and equations of motion as appropriate to the chosen semantics); see \cite{Collins:FoPQCD}.
The symmetry content of $\mathsf{Op}^{(n)}$ includes gauge invariance and (for massless QCD) flavor $SU(n_f)$.
In particular, the flavor decomposition splits $\mathsf{Op}^{(n)}$ into non-singlet blocks (which do not mix with gluons)
and a singlet block (in which the quark singlet mixes with the gluon operator).
Thus the admissible finite counterterms are block-diagonal with respect to this decomposition.

\paragraph{Moment-space interface algebra.}
Fix an accuracy datum $\alpha$ that includes a maximum spin $n\le n_{\max}$ (or, equivalently, restrict attention to finitely many moments).
Define the truncated operator sector
\[
\mathsf{Op}^{(\le n_{\max})}:=\bigoplus_{2\le n\le n_{\max}} \mathsf{Op}^{(n)}.
\]
Let $G$ denote the symmetry we insist counterterms respect (at minimum, gauge and flavor symmetry at the chosen truncation).
Then the interface algebra at this accuracy is
\begin{equation}\label{eq:Acoll-moment}
A_{\mathrm{coll}}^{(\le n_{\max})}\;:=\;\End_G\!\bigl(\mathsf{Op}^{(\le n_{\max})}\bigr)
\;\cong\;\bigoplus_{2\le n\le n_{\max}}\End_G(\mathsf{Op}^{(n)}),
\end{equation}
and the usual OPE mixing matrices $Z^{(n)}$ are elements of the unit group $\bigl(A_{\mathrm{coll}}^{(\le n_{\max})}\bigr)^\times$.
Concretely, in the unpolarized case the non-singlet blocks contribute scalar units (one per independent non-singlet channel),
while the singlet block contributes a $2\times 2$ matrix algebra acting on the $(\Sigma,g)$ sector.

\paragraph{$x$-space realization as convolution kernels.}
Returning to $x$-space, the same finite renormalizations are implemented by matrices of distribution-valued kernels
$Z=(Z_{ab})$ acting by Mellin convolution, with composition given by matrix convolution.
Accordingly, we define $A_{\mathrm{coll}}$ as the convolution-kernel realization of the commutant algebra in \eqref{eq:Acoll-moment}:
\begin{equation}\label{eq:Acoll}
(Z_1\cdot Z_2)_{ab}=\sum_c Z_{1,ac}\mconv Z_{2,cb},\qquad 1_{ab}=\delta_{ab}\,\delta(1-x).
\end{equation}
(Here $a,b$ label the chosen flavor/gluon basis, e.g.\ non-singlet and singlet channels.)
Invertible elements of $A_{\mathrm{coll}}$ are precisely the admissible finite scheme transformations used in collinear factorization
\cite{Collins:FoPQCD,PDG2024StructureFunctions}.

\begin{proposition}[Finite scheme transformations are units, and the action is faithful]\label{prop:units-faithful}
Work in a monoidal model $\cV$ in which $A_{\mathrm{coll}}$ is realized as a matrix algebra of convolution kernels with unit
$1=\delta(1-x)\,I$ and in which PDFs form a free left module $\vf$ over this algebra (e.g.\ a free module of flavor multiplets of
distributions on $(0,1)$ to the stated perturbative accuracy).
Then:
\begin{enumerate}[leftmargin=2.0em]
\item Any admissible finite factorization-scheme change in collinear DIS is implemented by an element $Z\in A_{\mathrm{coll}}^\times$,
acting by $\vf\mapsto Z\cdot \vf$ and $\vC_i\mapsto \vC_i\cdot Z^{-1}$.
Conversely, any $Z\in A_{\mathrm{coll}}^\times$ defines such a scheme change.
\item The induced left action of $A_{\mathrm{coll}}$ on $\vf$ is faithful: if $Z\cdot \vf=\vf$ for all $\vf$, then $Z=1$.
Similarly, the right action on coefficient multiplets is faithful in the same sense.
\end{enumerate}
\end{proposition}

\begin{proof}
(1) In the OPE/moment-space formulation, a finite renormalization corresponds to a change of basis in each $\mathsf{Op}^{(n)}$
by an invertible $G$-equivariant map, i.e.\ an element of $\End_G(\mathsf{Op}^{(n)})^\times$.
Assembling the finitely many spins $2\le n\le n_{\max}$ yields an element of $\bigl(A_{\mathrm{coll}}^{(\le n_{\max})}\bigr)^\times$,
and upon inverse Mellin transformation this becomes an invertible kernel matrix $Z\in A_{\mathrm{coll}}^\times$ acting by convolution.
The compensating transformation of Wilson coefficients (or coefficient functions) is exactly the inverse action, ensuring invariance
of the composite, as in \eqref{eq:scheme}; see \cite{Collins:FoPQCD,PDG2024StructureFunctions}.
Conversely, any $Z\in A_{\mathrm{coll}}^\times$ defines such a change of renormalized operator basis (or subtraction convention)
and hence a factorization-scheme change.

(2) If $Z\cdot \vf=\vf$ for all flavor multiplets $\vf$, apply this to the basis vectors $\vf=e_b\,\delta(1-x)$ (supported at $x=1$ in the
chosen monoidal model), where $e_b$ is the $b$th standard basis vector in flavor space. Then
$(Z\cdot \vf)_a = Z_{ab}$ must equal $\delta_{ab}\,\delta(1-x)$ for all $a,b$, i.e.\ $Z=1$.
The coefficient-side statement is analogous.
\end{proof}

\subsection{Coefficient and PDF modules}
Let $\vf$ be the flavor multiplet of PDFs and $\vC_i$ the multiplet of coefficient functions for $F_i$.
Define module actions
\[
(Z\cdot \vf)_a=\sum_b Z_{ab}\mconv f_b,\qquad (\vC_i\cdot Z)_b=\sum_a C_i^a\mconv Z_{ab}.
\]
Then \eqref{eq:scheme} is precisely the action of $Z\in A_{\mathrm{coll}}^\times$ and its inverse.

\subsection{Balancedness of the physical pairing}
Let $O$ be the object in $\cV$ representing scalar structure functions on $(0,1)$ (at fixed $Q$) to leading power.
Define
\begin{equation}\label{eq:Phi}
\Phi_i:\ \vC_i\otimes \vf \to O,\qquad \Phi_i(\vC_i,\vf):=\sum_a C_i^a\mconv f_a.
\end{equation}

\begin{proposition}[Scheme invariance implies balancedness]\label{prop:bal-dis}
$\Phi_i$ is $A_{\mathrm{coll}}$-balanced.
\end{proposition}

\begin{proof}
For $Z\in A_{\mathrm{coll}}$,
\[
\Phi_i((\vC_i\cdot Z),\vf)=\sum_{a,b}(C_i^a\mconv Z_{ab})\mconv f_b
=\sum_{a,b} C_i^a\mconv (Z_{ab}\mconv f_b)=\Phi_i(\vC_i,(Z\cdot \vf)),
\]
using associativity of $\mconv$.
\end{proof}

By Theorem~\ref{thm:core}, $\Phi_i$ factors uniquely through the core:
\[
\vC_i\otimes \vf \xrightarrow{q} \vC_i\otimes_{A_{\mathrm{coll}}}\vf \xrightarrow{\ \overline{\Phi}_i\ } O.
\]
We interpret $\vC_i\otimes_{A_{\mathrm{coll}}}\vf$ as the canonical scheme-invariant carrier of the leading-power content of $F_i$, and $\overline{\Phi}_i$ as the specialization of the universal lift \eqref{eq:coremap} to $A_\alpha=A_{\mathrm{coll}}$.

\subsection{Toy computation: a two-channel (singlet--gluon) moment-space core}\label{sec:toy}
The abstract quotient $\,\vC_i\otimes_{A_{\mathrm{coll}}}\vf\,$ becomes especially transparent after passing to moment space.
Concretely, apply the $n$th Mellin moment to the $x$-space pairing \eqref{eq:Phi}.
Mellin convolution becomes ordinary multiplication, so the recomposed observable for fixed spin/moment $n$ takes the familiar OPE form
\begin{equation}\label{eq:toy-obs}
F^{(n)} \;=\; \sum_{a} C^{(n)}_{a}\, f^{(n)}_{a},
\end{equation}
where $f^{(n)}_a$ are hadronic matrix elements (moments of PDFs) in the twist-two sector and $C^{(n)}_a$ are the corresponding Wilson-coefficient moments.

To connect the toy algebra $A$ directly to the physics of operator mixing, recall that in unpolarized DIS the only \emph{nontrivial} mixing at fixed $n$
occurs in the singlet sector: the quark singlet operator (for two active flavors, $\Sigma=u+d$) mixes with the gluon operator.
Thus, restricted to this $(\Sigma,g)$ block, a finite scheme change at fixed $n$ is represented by an invertible $2\times 2$ matrix
$Z^{(n)}\in GL_2(\mathbb{R})$, acting by
\[
f^{(n)}\mapsto f^{(n)\prime}=Z^{(n)} f^{(n)},\qquad
C^{(n)}\mapsto C^{(n)\prime}=C^{(n)}(Z^{(n)})^{-1},
\]
exactly as in the moment-space version of \eqref{eq:scheme}.
(We ignore the non-singlet channel, which would contribute an additional $1\times 1$ block and does not affect the point of the example.)

Accordingly, model the fixed-$n$ singlet block by taking $\cV=\mathrm{Vect}_{\mathbb{R}}$,
\[
A=\mathrm{Mat}_2(\mathbb{R}),\qquad
N=\mathbb{R}^{1\times 2}\ \text{(right $A$-module)},\qquad
M=\mathbb{R}^{2\times 1}\ \text{(left $A$-module)},
\]
with the usual row/column actions.
Write the coefficient moment as a row $C^{(n)}=(C^{(n)}_{\Sigma},C^{(n)}_{g})$ and the hadronic moment as a column
$f^{(n)}=(f^{(n)}_{\Sigma},f^{(n)}_{g})^{\mathsf T}$.
Then \eqref{eq:toy-obs} is simply the matrix contraction $F^{(n)}=C^{(n)}f^{(n)}$.
The balancing relation $(C\cdot Z)\otimes f\sim C\otimes (Z\cdot f)$ is precisely the statement that inserting a finite mixing matrix on either side
does not change $F^{(n)}$.

Because $N$ and $M$ are free rank-one modules over $A$, the relative tensor product collapses to a single copy of the ground field:
\begin{equation}\label{eq:toy-core}
N\otimes_A M \;\cong\; \mathbb{R},
\end{equation}
with isomorphism induced by matrix multiplication $(n,m)\mapsto nm$.
In this toy setting, the coequalizer quotient removes \emph{all} presentation-dependent information in $(C^{(n)},f^{(n)})$ and retains exactly the scheme-invariant
recomposed moment $F^{(n)}$.

As a numerical illustration (chosen for transparency and with a perturbatively small scheme change), take
\[
C^{(n)}=(1.2,0.3),\quad f^{(n)}=(0.4,0.6)^{\mathsf{T}},\quad
Z^{(n)}=\begin{pmatrix}1&0.1\\0&1\end{pmatrix}.
\]
Then $f^{(n)\prime}=Z^{(n)}f^{(n)}=(0.46,0.6)^{\mathsf{T}}$ and
$C^{(n)\prime}=C^{(n)}(Z^{(n)})^{-1}=(1.2,0.18)$, and indeed
\[
C^{(n)}f^{(n)} = 0.66 \;=\; C^{(n)\prime}f^{(n)\prime}.
\]
In perturbative QCD one typically has $Z^{(n)}=I+\mathcal{O}(\alpha_s)$ with entries fixed (order by order) by the finite parts of renormalization/mass-factorization counterterms,
but the balanced-quotient logic is identical: the core stores exactly the scheme-invariant content and nothing else.

\section{Systematic refinement beyond leading power}\label{sec:filtered}

Factorization/OPE statements are asymptotic in $Q$ and are naturally organized by twist/power.
A categorical way to make the ``no information loss at stated accuracy'' principle precise, while retaining a parsimonious core,
is to treat the power expansion as a filtration and compute cores levelwise.

\subsection{Filtered objects and levelwise cores}
Let $\mathrm{Fil}(\cV)$ denote a category of filtered objects in $\cV$ with filtration $\{F^{\le p}X\}_{p\ge 0}$.
Assuming $\cV$ has the relevant colimits levelwise, $\mathrm{Fil}(\cV)$ is cocomplete and the monoidal product is defined levelwise
(e.g.\ via the Day convolution or via a Rees construction when appropriate).

\begin{proposition}[Cores in filtered categories]\label{prop:fil}
Suppose Assumption~\ref{ass:ambient} holds in $\cV$ and that filtered colimits/coequalizers in $\mathrm{Fil}(\cV)$ are computed levelwise.
Let $A$ be a filtered algebra object and $(N,M)$ filtered right/left modules. Then the filtered relative tensor product exists and is computed levelwise:
\[
(F^{\le p}N)\otimes_{F^{\le p}A}(F^{\le p}M)\;\cong\;F^{\le p}(N\otimes_A M),
\]
and hence there is a tower of ``cores'' refining with $p$.
\end{proposition}

\begin{proof}
Form the defining coequalizer levelwise in each filtration degree and use the assumed levelwise preservation of the relevant coequalizers.
\end{proof}

\begin{remark}[DIS interpretation of the filtration]\label{rem:fil-dis}
For inclusive DIS, a natural choice is the twist/power filtration coming from the OPE:
one may take $F^{\le p}$ to mean ``twist $\le 2+p$'' (so $p=0$ is the leading-twist sector, $p=2$ corresponds to including twist~4, etc.).
This produces a tower of interface algebras $A^{\le p}$ (finite mixing within twist $\le 2+p$) and corresponding cores
$\mathrm{Core}^{\le p}:=C^{\le p}\otimes_{A^{\le p}} f^{\le p}$.
Increasing $p$ systematically adjoins the higher-twist operator/correlator sectors required by the OPE and refines the leading-twist core without altering the universal property at each level.
\end{remark}

If $O_R$ denotes an untruncated observable object and $\mathcal{R}_{>\alpha}\subset O_R$ the remainder subobject,
then $O_{R,\alpha}=O_R/\mathcal{R}_{>\alpha}$ is the semantic target used throughout.
Strengthening $\alpha$ amounts to shrinking $\mathcal{R}_{>\alpha}$, which produces a compatible inverse system of targets and hence a compatible inverse system of core maps
\eqref{eq:coremap}. Proposition~\ref{prop:fil} is one formal way to organize this refinement.

\section{Context, relations, and limitations}\label{sec:context}

\subsection{Interaction with SCET and resummation}
In soft-collinear effective theory (SCET), factorization and resummation are organized by matching QCD onto effective operators and evolving via RG equations.
From the perspective of this paper, SCET provides a particularly explicit construction of the interface algebra and module actions:
finite counterterms correspond to finite renormalizations of effective operators, while changes of representation (e.g.\ moment/Laplace/Fourier space) are monoidal transforms.
See e.g.\ the SCET lecture notes \cite{BecherBroggioFerrogliaSCET} and the PDG effective-theory review \cite{PDG2024SCET}.
Our theorem does \emph{not} prove SCET factorization; it can be used \emph{after} a factorized SCET formula is established to extract the scheme-invariant core.

\subsection{Non-perturbative renormalization and lattice matching}
The interface-algebra viewpoint applies equally well when renormalization factors are computed non-perturbatively (e.g.\ RI/MOM-type schemes) and then matched to $\MSbar$:
different renormalization conventions are simply different points in the same scheme groupoid, and the core object is invariant under those changes.
Modern lattice extractions of $x$-dependent structure via LaMET are a natural setting where one must manage such scheme transitions systematically; see e.g.\ \cite{Lin2025LatticePDFReview,JiEtAlHybridRenorm}.

\subsection{Renormalons and intrinsic non-perturbative ambiguities}
At fixed perturbative truncation, ``scheme'' can also refer to choices that reshuffle asymptotic series behavior.
Infrared renormalons encode genuine sensitivity to long distances and are tied to power corrections and non-perturbative ambiguities \cite{BenekeRenormalons}.
Our construction isolates redundancy \emph{within} a fixed truncation/semantic target; it does not remove intrinsic ambiguities that live in the remainder class $\mathcal{R}_{>\alpha}$.

\subsection{Relation to Hopf-algebraic renormalization and factorization algebras}
Hopf-algebraic formulations of perturbative renormalization (Connes--Kreimer and successors) organize counterterms and RG structure globally at the level of Feynman graphs
\cite{ConnesKreimer}. Our focus is different: we isolate the \emph{relative} scheme redundancy of a factorization presentation and compress it via a balanced quotient.
Nonetheless, the ``core as coinvariants'' viewpoint can be seen as a descent-style complement to Hopf-algebraic renormalization.

Factorization algebras provide another categorical framework for organizing observables and OPE structures in QFT \cite{CostelloGwilliamFA,WangWilliamsTH}.
Exploring the precise relationship between (i) local-to-global factorization-algebra semantics and (ii) the interface-algebra/balanced-quotient semantics used here
is an interesting direction, but lies beyond the scope of this collinear/OPE-focused paper.

\section{Conclusion}\label{sec:conclusion}
This paper isolates a structural feature that permeates perturbative QCD factorization and the
leading-twist OPE: the factorized constituents are not individually physical.
Short-distance coefficients and long-distance correlators are only defined up to admissible finite
redefinitions induced by collinear subtractions and renormalized-operator mixing, while the
recomposed observable is scheme independent (to a stated accuracy).
Our main contribution is to turn this familiar statement into a universal construction.

We encoded admissible finite redefinitions as an \emph{interface algebra object} $A_\alpha$
(Definition~6.1/6.3), acting on coefficient data $C$ as a right module and on hadronic data $f$
as a left module.
In this language, ``scheme invariance'' becomes a precise algebraic condition:
the physical evaluation map $\Phi:C\otimes f\to O_{R,\alpha}$ must be $A_\alpha$-balanced, i.e.\
inserting an admissible counterterm on the coefficient side is equivalent to inserting it on the
correlator side.
The Core Representation Theorem (Theorem~5.1) then states that the functor of balanced
pairings is \emph{representable} by the relative tensor product
\[
  \mathrm{Core}_\alpha \;:=\; C\otimes_{A_\alpha} f,
\]
and, moreover, that $\mathrm{Core}_\alpha$ is \emph{terminal} among all quotients of the naive composite
$C\otimes f$ that preserve scheme-invariant observable content.
This terminality is the precise mathematical sense in which the core is \emph{parsimonious}:
it compresses away exactly the internal scheme redundancy and nothing else.
Equivalently, $\mathrm{Core}_\alpha$ is an irreducible carrier of scheme-invariant information in the
universal-property sense: any further compression would necessarily erase information detectable
by some scheme-invariant observable.

The theorem is not merely a repackaging of the factorized formula; it provides a \emph{checklist}
for organizing, comparing, and extending factorized descriptions in a way that cleanly separates
physical content from presentation artifacts.
A practical workflow suggested by the paper is:
\begin{enumerate}
\item \textbf{Fix the semantics.} Choose a kinematic regime $R$ and accuracy datum $\alpha$
(power/twist truncation and perturbative order), thereby fixing a semantic target $O_{R,\alpha}$
(``agreement up to $\alpha$'').
\item \textbf{Identify the relevant operator sector.} Use locality/OPE (or an EFT matching description)
to select the truncated operator/correlator sector $\mathrm{Op}_\alpha$ controlling observables at $(R,\alpha)$.
\item \textbf{Determine the admissible scheme freedom.} Impose the symmetry constraints (gauge, flavor,
Lorentz/twist blocks, discrete symmetries) to identify which finite counterterms/mixing maps are allowed.
These assemble into the interface algebra $A_\alpha$ (often expressible as a commutant $\End_G(\mathrm{Op}_\alpha)$).
\item \textbf{Organize data as modules.} Treat coefficient data and long-distance data as right/left $A_\alpha$-modules.
This makes scheme transformations and operator mixing into honest algebra actions rather than ad hoc rules.
\item \textbf{Check balancedness once.} Verify the physical evaluation map $\Phi$ is $A_\alpha$-balanced;
in collinear DIS this is exactly associativity of convolution combined with the standard scheme-invariance identity.
\item \textbf{Work with the core object.} Replace the presentation-dependent composite $C\otimes f$ by the core
$C\otimes_{A_\alpha}f$, and treat $\overline{\Phi}:\,C\otimes_{A_\alpha}f\to O_{R,\alpha}$ as the canonical observable map.
Any other scheme-invariant pairing factors uniquely through the same core.
\end{enumerate}
This is an explicit conceptual shift with concrete benefits: it gives a canonical object that can be used to
compare different schemes, to change representations (moment space, $x$-space, Laplace/Fourier) without ambiguity
(Proposition~5.4), and to modularize the role of symmetry and truncation in determining what is ``allowed.''

In practical applications, one often starts from a small ``primitive'' set of long-distance operators/correlators
selected by symmetry and power counting, and then faces the question of what additional structures are
\emph{forced} by renormalization mixing.
Proposition~6.4 answers this in a clean way: the minimal sector stable under admissible redefinitions is the
$A_\alpha$-closure $\langle G_0\rangle_{A_\alpha}$.
This closure principle makes the frequently informal statement ``include all operators that mix at this order''
into a precise minimality claim.
Moreover, Section~8 explains how to treat power corrections as filtrations and obtain a tower of cores,
so that increasing accuracy refines the semantic target and refines the core levelwise, without changing the
universal property at each level.

Although we instantiated the construction in inclusive DIS to keep the discussion concrete, the categorical content
is intentionally agnostic to the analytic form of recomposition (Section~5.2):
the theorem does not require Mellin convolution per se, only an associative recomposition calculus equipped with an
interface-algebra action and the existence of the corresponding balanced quotient.
For physicists, this means the result should be read as a \emph{template}:
whenever a factorization theorem exhibits a scheme freedom that acts on the factors but cancels in the observable,
there is a canonical scheme-invariant core obtained by quotienting by the balancing relation.

The point of isolating the core object is not aesthetic economy but \emph{semantic control}:
it provides a canonical, scheme-invariant interface between perturbatively computed short-distance data and whatever
representation one uses for long-distance structure.
This is particularly valuable in modern hadron phenomenology, where long-distance information is often represented
numerically (global fits, lattice-matched objects, and other surrogates) and must be composed consistently with
perturbative ingredients while remaining invariant under permissible scheme changes.
The present paper supplies the categorical foundation for treating such long-distance objects as abstract elements
of a compositional calculus: once the $A_\alpha$-module action and the balanced evaluation are specified,
scheme invariance is enforced by construction via the universal core $C\otimes_{A_\alpha}f$.

At the same time, we stress the scope: we do not prove factorization, nor do we remove intrinsic physical
ambiguities living in the remainder beyond the chosen semantic target (Section~9.3).
Rather, the result is a structural tool that becomes applicable once a factorized/OPE description is established,
and that can guide systematic refinement (via filtrations/closures) and meaningful comparison of presentations.
We expect this viewpoint to be useful wherever new factorization statements are sought: it clarifies what must be
added to a model (new interfaces/sectors) to restore representability of the semantics, and it provides a canonical
notion of ``no further compression without loss of invariant content'' at fixed $(R,\alpha)$.

\section*{Acknowledgments}
This work was funded from the budget of the University of Virginia. No additional grants have been received to conduct or
direct this particular study.

Conflict of Interest: The authors of this work declare that they have
no conflicts of interest.

\bibliographystyle{unsrt}
\bibliography{ref}

\end{document}